\documentclass[journal,twoside,web]{ieeecolor}
\usepackage{lcsys}
\usepackage{cite}
\usepackage{amsmath,amssymb,amsfonts,amsthm}
\usepackage{graphicx}
\usepackage{textcomp}
\usepackage{etoolbox}
\usepackage{url}
\newcommand{\repourl}{\url{https://github.com/das-goutam/mppi-quantum-query}}
\apptocmd{\thebibliography}{\setlength{\itemsep}{0pt}}{}{}

\theoremstyle{plain}
\newtheorem{theorem}{Theorem}
\newtheorem{lemma}{Lemma}
\newtheorem{proposition}{Proposition}
\newtheorem{corollary}{Corollary}
\theoremstyle{definition}
\newtheorem{assumption}{Assumption}
\newtheorem{problem}{Problem}
\theoremstyle{remark}
\newtheorem{remark}{Remark}

\newcommand{\E}{\mathbb{E}}
\newcommand{\PP}{\mathbb{P}}
\newcommand{\Rset}{\mathbb{R}}
\newcommand{\ind}{\mathbf{1}}
\newcommand{\TV}{\operatorname{TV}}
\newcommand{\calA}{\mathcal{A}}
\newcommand{\calC}{\mathcal{C}}
\newcommand{\calZ}{\mathcal{Z}}
\newcommand{\dt}{\Delta t}
\newcommand{\eps}{\varepsilon}

\makeatletter
\def\fnum@figure{{\color{subsectioncolor}\sf Fig.~\thefigure}}
\def\fnum@table{{\color{subsectioncolor}\sf TABLE~\thetable}}
\patchcmd{\ps@titlepagestyle}{\includegraphics[width=\logowidth]{\logoname.eps}}{}{}{}
\patchcmd{\ps@titlepagestyle}{\includegraphics[width=\logowidth]{\logoname.eps}}{}{}{}
\patchcmd{\ps@titlepagestyle}{\includegraphics[width=\logowidth]{\logoname.eps}}{}{}{}
\patchcmd{\ps@titlepagestyle}{\includegraphics[width=\logowidth]{\logoname.eps}}{}{}{}
\makeatother

\renewcommand{\baselinestretch}{0.96}
\begin{document}

\title{Model Predictive Path Integral Control as a Quantum Query Problem}
\author{Goutam Das, \IEEEmembership{Member, IEEE}, and Takashi Tanaka, \IEEEmembership{Senior Member, IEEE}
\thanks{The authors are with the School of Aeronautics and Astronautics,
Purdue University, West Lafayette, IN 47907 USA (e-mail:
das347@purdue.edu; ttanaka16@purdue.edu).}}

\maketitle
\begin{abstract}
Model predictive path integral (MPPI) control computes its inputs as cost-weighted averages over simulated trajectories. However, it typically requires many simulations, especially when low-cost trajectories are rare. We therefore formulate this computation as a quantum query problem. By replacing Gaussian input perturbations with random signs, we express the MPPI update as a ratio of averages over a finite set of perturbation sequences. We then construct quantum circuits that encode these averages as measurement probabilities, which allows quantum amplitude estimation to compute the update with fewer queries than classical sampling. Numerical examples confirm the predicted query scalings and show that a quantum advantage requires many perturbation sequences.
\end{abstract}

\section{Introduction}\label{sec:intro}

\IEEEPARstart{P}{ath} integral control expresses the optimal control of a class of stochastic systems as an expectation over trajectories of the uncontrolled system~\cite{kappen05,todorov09}. Model predictive path integral (MPPI) control applies this idea in a receding-horizon manner~\cite{williams17,williams18}. At each step, it simulates many trajectories with randomly perturbed inputs and applies the cost-weighted average of their first inputs. MPPI typically needs many simulated trajectories to estimate this average accurately, especially when few trajectories have low cost~\cite{yoon22,patil24}.

In a \textit{quantum query problem}, an algorithm can access its input only through a subroutine called an \textit{oracle}. The number of oracle calls the algorithm makes is called its \textit{query complexity}~\cite[Ch.~6]{nielsenchuang}. For estimating an average, \textit{quantum amplitude estimation}~\cite{bhmt02} needs quadratically fewer queries than classical Monte Carlo sampling~\cite{montanaro15}, and no other quantum algorithm can do better in general~\cite{nayakwu99}. Since the MPPI update is a cost-weighted average of simulated trajectories, we study MPPI as a quantum query problem. The same approach applies in principle to other sampling-based updates of this form, such as the mean update of the Gaussian cross-entropy method~\cite{deboer05}, which we do not analyze here.

However, applying amplitude estimation to MPPI poses three difficulties. First, preparing a quantum state that represents the Gaussian perturbation distribution adds computational cost to each query. This added cost can offset the benefit of using fewer queries~\cite{herbert21}. Second, each trajectory and its cost must be computed accurately inside the quantum circuit, which is difficult because of rounding errors caused by the limited number of qubits. Third, the reduction in queries must outweigh the cost of simulating trajectories on a quantum computer, since quadratic speedups become useful only for large problems~\cite{babbush21}. This paper addresses these issues by showing how the MPPI update can be formulated as a quantum query problem, determining its query complexity, and identifying how many perturbation sequences are needed for a fault-tolerant quantum computer to outperform a classical computer.

Quantum methods have also been proposed for nonlinear model predictive control~\cite{novara24} and reinforcement learning~\cite{weinberg26}, but, to our knowledge, not for the finite-horizon trajectory averages that determine the MPPI input. Our contributions are as follows.

\begin{itemize}
\item[\textbf{(C1)}] We formulate one MPPI update as a quantum query problem and construct the corresponding quantum circuits (Sections~\ref{sec:pic:disc} and~\ref{sec:quantum}).
\item[\textbf{(C2)}] We prove that amplitude estimation computes this update with quadratically fewer queries than classical sampling (Section~\ref{sec:results:qae}).
\item[\textbf{(C3)}] We show that, when the lowest-cost sequence is unique, the MPPI update converges to its first input as the weights increasingly favor lower costs. Quantum minimum finding~\cite{durr96} identifies this sequence with quadratically fewer queries than classical search (Section~\ref{sec:results:lowtemp}).
\item[\textbf{(C4)}] We derive a \textit{crossover condition}, under which the modeled quantum cost falls below the classical cost, and show that it requires a large number of perturbation sequences. We evaluate this condition on a linear and a nonlinear example (Sections~\ref{sec:results:cost} and~\ref{sec:numerics}).
\end{itemize}

\section{Path Integral Control}\label{sec:pic}

This section reviews path integral control and then states the MPPI update in the form that we encode in quantum circuits.

\subsection{Control problem and path-integral solution}\label{sec:pic:soc}

Consider the stochastic differential equation on $[0,T]$,
\begin{equation}\label{eq:sde}
    dx_t=f(x_t)\,dt+G(x_t)\bigl(u_t\,dt+\sigma\,dw_t\bigr),\qquad x_0=\bar{x},
\end{equation}
where $x_t\in\Rset^n$, $u_t\in\Rset^m$, $w_t$ is an $m$-dimensional standard Brownian motion, $G(x)\in\Rset^{n\times m}$, and $\sigma\in\Rset^{m\times m}$ is invertible. The finite-horizon cost is
\begin{equation}\label{eq:cost}
    J(u)=\E\!\left[\phi(x_T)+\int_0^T\left(q(x_t)+\tfrac12u_t^{\!\top}Ru_t\right)dt\right],
\end{equation}
where $R\in\Rset^{m\times m}$, $q:\Rset^n\to[0,\infty)$ is the running state cost, and $\phi:\Rset^n\to[0,\infty)$ is the terminal cost. The value function $V(x,t)$ is the optimal cost-to-go from $x$ at time $t$.

\begin{assumption}\label{as:soc}
\emph{(A1)} The control and diffusion enter through the same channels $G$. \emph{(A2)} $R\succ0$. \emph{(A3)} The matching condition $\sigma\sigma^{\!\top}=\lambda R^{-1}$ holds for some $\lambda>0$.
\end{assumption}

Assumption~\ref{as:soc} is standard in path integral control~\cite{kappen05,todorov09,williams17}. Under this assumption and standard regularity conditions, the Cole--Hopf transformation $V=-\lambda\log\psi$ turns the nonlinear Hamilton--Jacobi--Bellman equation into a linear one in terms of the desirability function $\psi(x,t)$~\cite{todorov09}. By the Feynman--Kac formula,
\begin{equation}\label{eq:fk}
\psi(x,t)=\E_{\PP}\!\left[e^{-(\phi(x_T)+\int_t^Tq(x_s)\,ds)/\lambda}\Bigm|x_t=x\right],
\end{equation}
where $\PP$ is the probability distribution of the trajectories of the uncontrolled system $dx_t=f(x_t)\,dt+G(x_t)\sigma\,dw_t$. The optimal control is given by $u^*=\lambda R^{-1}G^{\!\top}\nabla\log\psi$.

\subsection{MPPI update with random signs}\label{sec:pic:disc}

To apply path integral control in practice, MPPI~\cite{williams17,williams18} discretizes~\eqref{eq:sde} over $N$ time steps of length $\dt=T/N$ with the Euler--Maruyama scheme~\cite[Sec.~9.1]{kloedenplaten},
\begin{equation}
x_{k+1}=x_k+f(x_k)\dt+\sigma_s\sqrt{\dt}\,G(x_k)\,\xi_k,\qquad x_0=\bar x,
\end{equation}
where $\sigma$ is replaced by $\sigma_sI_m$ for a sampling scale $\sigma_s>0$, and $\xi_0,\ldots,\xi_{N-1}$ are independent standard Gaussian vectors in $\Rset^m$. Equivalently, the zero nominal input is perturbed by $\delta u_k=\sigma_s\xi_k/\sqrt{\dt}$. For $K$ simulated trajectories $j=1,\ldots,K$ with state costs $S^{(j)}$, each the sum of the running costs $q(x_k)\dt$ and the terminal cost $\phi(x_N)$, MPPI applies the weighted average of the first inputs,
\begin{equation}
u_0=\sum_{j=1}^K\omega^{(j)}\,\delta u_0^{(j)},\qquad \omega^{(j)}=\frac{e^{-S^{(j)}/\lambda}}{\sum_{l=1}^Ke^{-S^{(l)}/\lambda}}.
\end{equation}
Lower-cost trajectories receive exponentially larger weights, and a smaller $\lambda$ makes this preference stronger. MPPI recomputes this update from the measured state at every control step, and we analyze one such update.

Our construction replaces each Gaussian vector $\xi_k$ with a vector $z_k\in\Rset^m$ whose entries $z_{k,i}$ are independent random signs, or Rademacher variables, each equal to $-1$ or $+1$ with equal probability, so that $\delta u_k=\sigma_sz_k/\sqrt{\dt}$. Such binary perturbations are standard in the weak numerical solution of stochastic differential equations~\cite[Sec.~14.1]{kloedenplaten} and in simultaneous perturbation stochastic approximation, where they are the optimal perturbation distribution~\cite{sadegh98}. They have the same mean and variance as the entries of $\xi_k$, and under standard smoothness conditions on the dynamics and costs, the resulting simplified weak Euler scheme approximates expectations of functions of the trajectory with the same first-order accuracy in $\dt$ as the Gaussian scheme, while making the set of perturbation sequences finite. Each perturbation sequence is therefore a \emph{sign sequence} $z=(z_0,\ldots,z_{N-1})$ of $n_q:=Nm$ signs, and there are only $D:=2^{n_q}$ such sequences. We denote their set by $\calZ$, the uniform distribution on $\calZ$ by $\mu$, and the corresponding average by $\E_\mu$.

With $\xi_k$ replaced by $z_k$, the same scheme gives, for each sign sequence $z\in\calZ$, the deterministic trajectory and cost
\begin{align}
x_{k+1}(z)&=x_k(z)+f(x_k(z))\dt+\sigma_s\sqrt{\dt}\,G(x_k(z))\,z_k,\label{eq:euler}\\
S(z)&=\phi(x_N(z))+\textstyle\sum_{k=0}^{N-1}q(x_k(z))\dt,\label{eq:S}
\end{align}
where the trajectory \eqref{eq:euler} starts from $x_0(z)=\bar{x}$. Since $q$ and $\phi$ are nonnegative, $S(z)\ge0$.

Using all $D$ sequences in place of $K$ random samples, and with $g(z):=e^{-S(z)/\lambda}$, the MPPI update becomes
\begin{equation}\label{eq:update}
    \omega_\lambda(z):=\frac{g(z)}{\sum_{z'}g(z')},\quad u_0^{(\lambda)}:=\frac{\sigma_s}{\sqrt{\dt}}\sum_{z}\omega_\lambda(z)\,z_0.
\end{equation}
This is the MPPI update of~\cite{williams18}, evaluated over the finite set $\calZ$ with zero nominal input. To write~\eqref{eq:update} with averages, let
\begin{equation}
    a:=\E_{\mu}[g(z)],\qquad b_i:=\E_{\mu}\!\left[g(z)\ind\{z_{0,i}=+1\}\right],
\end{equation}
that is, $a$ is the average weight of all sequences, and $b_i$ is the part of this average contributed by the sequences whose first sign in coordinate $i$ is $+1$. Since $z_{0,i}=2\ind\{z_{0,i}=+1\}-1$, linearity gives $\E_{\mu}[g(z)z_{0,i}]=2b_i-a$. Substituting this into \eqref{eq:update} gives
\begin{equation}\label{eq:ratio}
    u_{0,i}^{(\lambda)}=\frac{\sigma_s}{\sqrt{\dt}}\left(\frac{2b_i}{a}-1\right).
\end{equation}

To relate the average weight $a$ to the desirability $\psi$, the sampling must match the noise in~\eqref{eq:sde}, that is, $\sigma=\sigma_sI_m$. For $R=\rho I_m$, A3 then requires $\lambda=\rho\sigma_s^2$. In this case, the first-order accuracy of the simplified weak Euler scheme gives $|a-\psi(\bar x,0)|\le C_\lambda\dt$ for a constant $C_\lambda$, so $a$ is the discrete counterpart of the desirability.

\begin{remark}\label{rem:range}
Since $S(z)\ge0$, we have $g(z)\in(0,1]$ and $0<b_i<a\le1$, so $a$ and $b_i$ can be encoded as probabilities. However, $a$ can be extremely small when all costs are large compared with $\lambda$, which makes it hard to represent with finitely many bits. Since~\eqref{eq:ratio} is a ratio, this can be avoided by shifting the cost. With $S^*:=\min_zS(z)$, subtracting any known constant $\underline S\in[0,S^*]$ from $S$ leaves~\eqref{eq:ratio} unchanged, keeps the weights in $(0,1]$, and multiplies $a$ by $e^{\underline S/\lambda}$. The shift $\underline S=S^*$ gives the normalized weights and their average
\begin{equation}
g'(z):=e^{-(S(z)-S^*)/\lambda},\;a':=\E_\mu[g']=a\,e^{S^*/\lambda}\ge\frac1D,
\end{equation}
and we call $a'$ the normalized desirability. Unlike $a$, it does not change when a constant is added to the cost. Rare low-cost trajectories thus correspond to small $a'$.
\end{remark}

\section{Quantum Oracle Formulation}\label{sec:quantum}

Equation~\eqref{eq:ratio} reduces the MPPI update to the averages $a$ and $b_i$ over the set $\calZ$. This section explains how quantum circuits compute these averages. Here, the oracles are reversible circuits that simulate the trajectory of a given sign sequence and compute its cost or weight. A register of $n_q$ qubits holds all $D$ sign sequences in superposition, a reversible implementation of~\eqref{eq:euler}--\eqref{eq:S} evaluates all of them simultaneously, and one additional qubit, called the \emph{output qubit}, is rotated so that measuring it returns $1$ with probability $a$ or $b_i$. Amplitude estimation estimates this probability with an error of order $1/Q$, with high probability, after $Q$ circuit calls, whereas the standard deviation of the error of classical sampling with $K$ simulations is of order $1/\sqrt K$. Estimators of this type are known as quantum Monte Carlo methods~\cite{montanaro15}.

An $n_q$-qubit register stores a sign sequence $z$, with the qubit for $z_{k,i}$ in state $|1\rangle$ if $z_{k,i}=+1$ and in state $|0\rangle$ if $z_{k,i}=-1$. A Hadamard gate on every qubit prepares the uniform superposition $D^{-1/2}\sum_{z\in\calZ}|z\rangle$, which corresponds to $\mu$. Amplitude estimation does not draw perturbation samples one at a time. Instead, it uses this superposition, which contains every possible sign sequence with the probability $1/D$ given by $\mu$. Preparing it takes only $n_q$ Hadamard gates.

\emph{Cost oracle:} The trajectory~\eqref{eq:euler} and cost~\eqref{eq:S} are evaluated with a fixed number of bits by reversible adders and multipliers~\cite{vedral96,haner18}, with the intermediate states $x_k(z)$ held in work registers. Within each oracle call, these registers are reset to zero right after $S(z)$ is computed, by running these operations backwards (uncomputation~\cite[Sec.~3.2.5]{nielsenchuang}). The resulting circuit performs the same computation as one classical simulation and maps each sequence to its cost,
\begin{equation}\label{eq:OS}
O_S|z\rangle|0^w\rangle=|z\rangle|S(z)\rangle,
\end{equation}
where the second register consists of $w$ qubits, which hold $S(z)$ as a $w$-bit binary number after the call.

\emph{Threshold oracle:} For a given threshold $\gamma$, a reversible comparison after the cost oracle marks each sequence with $S(z)<\gamma$ by multiplying its amplitude by $-1$, as in Grover search~\cite[Ch.~6]{nielsenchuang}. Uncomputing the cost gives the oracle used by minimum finding, $O_\gamma|z\rangle=(-1)^{\ind\{S(z)<\gamma\}}|z\rangle$.

\emph{State preparation:} After $O_S$, the circuit computes the angle $\theta(z)=2\arcsin\sqrt{g(z)}=2\arcsin e^{-S(z)/(2\lambda)}$ into a register. Since the circuit can only add and multiply, this function is evaluated by a piecewise polynomial approximation~\cite{haner18}. The angle register controls the rotation $R_y(\theta):=e^{-i\theta Y/2}$ of the output qubit about the $y$-axis, with $Y$ the Pauli-$Y$ matrix. Since $R_y(\theta)|0\rangle=\cos\frac\theta2|0\rangle+\sin\frac\theta2|1\rangle$, these steps, followed by uncomputing the angle and cost registers, define the state-preparation circuit $\calA$ by
\begin{equation}\label{eq:Aprep}
    \calA|0^{n_q}\rangle|0\rangle
    :={}\frac{1}{\sqrt D}\sum_z|z\rangle
    \Bigl(\sqrt{1-g(z)}\,|0\rangle
    +\sqrt{g(z)}\,|1\rangle\Bigr).
\end{equation}
Let the output probability $p_{\calC}$ of a circuit $\calC$ denote the probability of measuring the output qubit in $|1\rangle$ after $\calC$. Then $p_{\calA}=D^{-1}\sum_zg(z)=a$. To encode $b_i$, we use a circuit $\calA_i$ that is identical to $\calA$ except that the rotation is additionally controlled by the qubit holding $z_{0,i}$, so that it acts only when $z_{0,i}=+1$. Hence, $p_{\calA_i}=D^{-1}\sum_zg(z)\ind\{z_{0,i}=+1\}=b_i$.
Since $a$ and $b_i$ determine the update~\eqref{eq:ratio}, $\calA$ and $\calA_i$ are the oracles for estimating it. Amplitude estimation uses them through amplitude-amplification steps, each of which applies the oracle once and its inverse once~\cite{bhmt02}. The query complexity therefore counts calls to $\calA$, $\calA_i$, and their inverses. An inverse has the same gate cost as the circuit itself, and $\calA_i$ costs about the same as $\calA$, since it adds only one control qubit to the rotation.

The circuits above assume exact computation. In practice, they use finitely many bits, so their output probabilities contain small errors, which Assumption~\ref{as:arith} bounds for a desired accuracy.

\begin{assumption}\label{as:arith}
The implemented circuits $\tilde\calA$ and $\tilde\calA_i$ have output probabilities $\tilde a:=p_{\tilde\calA}$ and $\tilde b_i:=p_{\tilde\calA_i}$ satisfying $|\tilde a-a|\le\eta a$ and $|\tilde b_i-b_i|\le\eta a$ for some $\eta\ge0$.
\end{assumption}

We denote by $e_S$ and $e_\theta$ the largest errors in the computed cost and rotation angle. The next lemma shows how small they must be, which determines how many bits are required.

\begin{lemma}\label{lem:prec}
Let $\eta\in(0,1]$. Suppose that, for every $z$, the circuit computes a cost $\tilde S(z)$ with $|\tilde S(z)-S(z)|\le e_S$ and an angle $\tilde\theta(z)$ with $|\tilde\theta(z)-\bar\theta(z)|\le e_\theta$, where $\bar\theta(z):=2\arcsin e^{-\max\{\tilde S(z),0\}/(2\lambda)}$ is the angle that corresponds to the computed cost. If $e_S\le\lambda\eta/4$ and $e_\theta\le\eta\sqrt a/4$, then Assumption~\ref{as:arith} holds.
\end{lemma}

\begin{proof}
Let $\tilde g:=\sin^2(\tilde\theta/2)$ be the encoded weight and $\bar g:=\sin^2(\bar\theta/2)=e^{-\max\{\tilde S,0\}/\lambda}$. Since $S\ge0$, $|\max\{\tilde S,0\}-S|\le e_S$, so with $x:=e_S/\lambda\le\frac\eta4\le\frac14$, $|\bar g-g|\le(e^{x}-1)g\le xe^xg\le\frac{\eta}{3}g$, and in particular $\bar g\le\frac43g$. For $\beta\in[0,\frac\pi2]$, the identity $\sin^2\alpha-\sin^2\beta=\sin(\alpha-\beta)\sin(\alpha+\beta)$ gives
\begin{equation}
|\sin^2\alpha-\sin^2\beta|\le|\alpha-\beta|\bigl(2\sin\beta+|\alpha-\beta|\bigr).
\end{equation}
With $\alpha=\tilde\theta/2$ and $\beta=\bar\theta/2$, so that $|\alpha-\beta|\le e_\theta/2$ and $\sin\beta=\sqrt{\bar g}\le1.16\sqrt g$, this gives $|\tilde g-\bar g|\le e_\theta\sqrt{\bar g}+e_\theta^2/4\le1.16\,e_\theta\sqrt g+e_\theta^2/4$. Adding the bound on $|\bar g-g|$ gives $|\tilde g-g|\le1.16\,e_\theta\sqrt g+e_\theta^2/4+\frac\eta3g$ for every $z$. Since $\tilde a=\E_\mu\tilde g$ and $a=\E_\mu g$, averaging over $\mu$ and using Jensen's inequality $\E_\mu\sqrt g\le\sqrt{\E_\mu g}=\sqrt a$ yields
\begin{equation}
|\tilde a-a|\le\E_\mu|\tilde g-g|\le1.16\,e_\theta\sqrt a+\frac{e_\theta^2}{4}+\frac\eta3a.
\end{equation}
With $e_\theta\le\eta\sqrt a/4$ and $\eta\le1$, the three terms are at most $\frac{1.16}{4}\eta a$, $\frac1{64}\eta a$, and $\frac13\eta a$, whose sum is less than $\eta a$.
The same bound holds for $\tilde b_i$, since the indicator only removes nonnegative terms.
\end{proof}

For the angle, rounding to $n_\theta$ bits on $[0,\pi]$ gives $e_\theta\le\pi2^{-n_\theta-1}$, so $n_\theta\ge\log_2(2\pi/(\eta\sqrt a))$ is sufficient, which grows only logarithmically as $a$ decreases. The cost error $e_S$ depends on the dynamics, so it can be measured by repeating the cost computation on a classical computer with the same number of bits.

Using these oracles, we consider the following two query problems.

\begin{problem}\label{prob:Q1}
Given query access to $O_S$ and a tolerance $\mathrm{tol}\ge0$, return a sequence $\hat z$ with $S(\hat z)\le S^*+\mathrm{tol}$ with probability at least $1-\delta$, where $\delta\in(0,1)$ is the failure probability.
\end{problem}

\begin{problem}\label{prob:Q2}
Given query access to $\calA$, $\calA_i$, and their inverses, return $\hat u_0\in\Rset^m$ with $\|\hat u_0-u_0^{(\lambda)}\|_\infty\le\eps\sigma_s/\sqrt{\dt}$ with probability at least $1-\delta$, where $\delta\in(0,1)$ is the failure probability and $\eps\in(0,1]$ is the accuracy.
\end{problem}

\section{Main Results}\label{sec:results}

This section bounds the number of queries for the update~\eqref{eq:ratio}, compares it with classical sampling, and examines when the quantum approach could have a lower total cost.

\subsection{Update estimation}\label{sec:results:qae}

In this subsection, $g$ denotes the weight after any cost shift allowed by Remark~\ref{rem:range}, and $a=\E_{\mu}[g]$.

The analysis uses two known guarantees for amplitude estimation, which we state without proof.

\begin{lemma}\label{lem:ae}
Let $\calC$ be a circuit with output probability $p$. (i)~For every integer $M\ge1$, the algorithm of~\cite[Thm.~12]{bhmt02} makes $O(M)$ calls to $\calC$ and $\calC^{-1}$ and returns $\hat p$ with $|\hat p-p|\le2\pi\sqrt{p(1-p)}/M+\pi^2/M^2$ with probability at least $8/\pi^2$. (ii)~If $p>0$, then for every $\eps',\delta'\in(0,1)$, the algorithm of~\cite{ar20} needs no prior knowledge of $p$ and, with probability at least $1-\delta'$, returns $\hat p$ with $|\hat p-p|\le\eps'p$ using $O(\log(1/\delta')/(\eps'\sqrt p))$ calls.
\end{lemma}

Since~\eqref{eq:ratio} is a ratio, we also bound how errors in $a$ and $b_i$ affect it.

\begin{lemma}\label{lem:ratio}
Let $a\in(0,1]$, $b\in[0,a]$, $\eps\in(0,1]$, and $\kappa>1$. If
$|\hat a-a|\le\eps a/\kappa$ and $|\hat b-b|\le\eps a/\kappa$, then $|2\hat b/\hat a-2b/a|\le4\eps/(\kappa-1)$.
\end{lemma}

\begin{proof}
Since $\eps\le1$, $|\hat a-a|\le a/\kappa$, so
$\hat a\ge\tfrac{\kappa-1}{\kappa}a>0$. Writing
\begin{equation}
\frac{\hat b}{\hat a}-\frac{b}{a}
=\frac{\hat b-b}{\hat a}+\frac{b}{a}\cdot\frac{a-\hat a}{\hat a}
\end{equation}
and using $b/a\le1$, each term is at most
$(\eps a/\kappa)\big/\!\bigl(\tfrac{\kappa-1}{\kappa}a\bigr)
=\eps/(\kappa-1)$. Summing and multiplying by two gives the result.
\end{proof}

Combining Lemmas~\ref{lem:ae} and~\ref{lem:ratio} gives the number of queries for the whole update.

\begin{theorem}[Quantum query complexity]\label{thm:qae}
Fix $N$, $\lambda>0$, an accuracy $\eps\in(0,1]$, and a constant $\kappa\ge5$. Suppose Assumption~\ref{as:arith} holds with $\eta\le\eps/(2\kappa+1)$. For any failure probability $\delta\in(0,1)$, there is a quantum algorithm that, with probability at least $1-\delta$,
solves Problem~\ref{prob:Q2}, returning $\hat u_0$ with
\begin{equation}
\|\hat u_0-u_0^{(\lambda)}\|_\infty\le\frac{4}{\kappa-1}\,\frac{\eps\sigma_s}{\sqrt{\dt}}\le\frac{\eps\sigma_s}{\sqrt{\dt}},
\end{equation}
using
\begin{equation}\label{eq:qcomplexity}
O\!\left(
\frac{m}{\eps\sqrt a}
\log\frac{m}{\delta}
\right)
\end{equation}
queries to $\calA$, $\calA_i$, and their inverses.
\end{theorem}

\begin{proof}
Let $c:=2\kappa+1$, so $\eta\le\eps/c$. As $4/(\kappa-1)\le1$, by Lemma~\ref{lem:ratio} it suffices to obtain $|\hat a-a|\le\eps a/\kappa$ and $|\hat b_i-b_i|\le\eps a/\kappa$ for all $i$ jointly with probability at least $1-\delta$. Lemma~\ref{lem:ae}(ii) with $\eps'=\eps/c$ and $\delta'=\delta/(m+1)$ gives, jointly with its query count and with probability at least $1-\delta/(m+1)$, $|\hat a-\tilde a|\le(\eps/c)\tilde a$, so, with $\tilde a\le(1+\eta)a$ and $\eta\le\eps/c\le1/c$, $|\hat a-a|\le\frac{\eps}{c}(1+\eta)a+\eta a\le\frac{\eps}{c}(2+\frac1c)a\le\eps a/\kappa$, since $c^2-2\kappa c-\kappa=\kappa+1>0$. Thus $\frac{\kappa-1}{\kappa}a\le\hat a\le\frac{\kappa+1}{\kappa}a$, and $\tilde a\ge(1-\frac1c)a$ gives cost $O(\log(\frac{m+1}{\delta})/(\eps\sqrt a))$. For each numerator, Lemma~\ref{lem:ae}(i) with
\begin{equation}
M=\biggl\lceil\frac{4\pi c}{\eps\sqrt{\hat a}}\sqrt{\frac{2\kappa}{\kappa-1}}+\pi\sqrt{\frac{2c}{\eps\hat a}}\,\biggr\rceil=O\Bigl(\frac{1}{\eps\sqrt a}\Bigr)
\end{equation}
bounds both error terms by $\eps\hat a/(2c)$, since $\tilde b_i\le b_i+\eta a\le2a\le\frac{2\kappa}{\kappa-1}\hat a$. Hence, $|\hat b_i-b_i|\le\frac{\eps}{c}\hat a+\eta a\le(\frac{\kappa+1}{\kappa}+1)\frac{\eps a}{c}=\eps a/\kappa$. On this event, the values of $M$ are fixed. Each run of the numerator estimator succeeds with probability at least $8/\pi^2>\frac12$, so medians of $O(\log((m+1)/\delta))$ runs fail with probability at most $\delta/(m+1)$ each, and a union bound and Lemma~\ref{lem:ratio} give the claim at cost~\eqref{eq:qcomplexity}.
\end{proof}

For comparison, the next proposition gives the number of simulations that classical sampling needs to solve the same problem. Since the ratio of sample averages does not change when all weights are multiplied by a constant, we state it with the normalized weights $g'$ of Remark~\ref{rem:range}.

\begin{proposition}\label{prop:classical}
Classical sampling solves Problem~\ref{prob:Q2} with $O(\log(m/\delta)/(\eps^2a'))$ simulations.
\end{proposition}
\begin{proof}
Let $b_i':=\E_\mu[g'(z)\ind\{z_{0,i}=+1\}]$. Draw $K$ independent sequences from $\mu$ and estimate $a'$ and every $b_i'$ by sample averages. Each simulation gives the weight $g'(z)$ and all signs $z_{0,i}$, so the same $K$ samples serve all $m+1$ averages. Each sampled quantity lies in $[0,1]$, so its variance is at most its mean, which is at most $a'$. For fixed $\kappa\ge5$, Bernstein's inequality and a union bound give additive error at most $\eps a'/\kappa$ in all $m+1$ averages with probability at least $1-\delta$ when $K=O(\log(m/\delta)/(\eps^2a'))$, and Lemma~\ref{lem:ratio} bounds the error of the ratio.
\end{proof}

\begin{remark}\label{rem:lower}
The dependence on $\eps$ and $a'$ in Theorem~\ref{thm:qae} after the shift by $S^*$ and in Proposition~\ref{prop:classical} matches known lower bounds for estimating an average. For example, if each cost is either $S^*$ or much larger than $S^*+\lambda$, every weight $g'(z)$ is close to $1$ or $0$. Then $a'$ is essentially the fraction of sequences with weight one, and estimating it to relative accuracy $\eps$ is the approximate counting problem. For $a'\le1/2$ and $\eps\le1/3$, approximate counting requires $\Omega(1/(\eps\sqrt{a'}))$ quantum queries~\cite{nayakwu99} and $\Omega(1/(\eps^2a'))$ classical samples.
\end{remark}

\subsection{Minimum-cost search}\label{sec:results:lowtemp}

In MPPI, $\lambda$ is a tuning parameter. As $\lambda$ decreases, the weights $\omega_\lambda(z)$ in~\eqref{eq:update} concentrate on the lowest-cost sequences. To study this regime, we fix the sign sequences and their costs and show that, as $\lambda\to0$, the update $u_0^{(\lambda)}$ converges to the average first input of the minimum-cost sequences. When the minimum-cost sequence is unique, computing~\eqref{eq:update} thus reduces to the search of Problem~\ref{prob:Q1}. Let $\calZ^*:=\arg\min_zS(z)$ and $\TV(\nu,\nu'):=\frac12\sum_z|\nu(z)-\nu'(z)|$.

\begin{lemma}\label{lem:conc}
Suppose $1\le |\calZ^*|<D$ and define the cost gap
$\Delta:=\min_{z\notin\calZ^*}\bigl(S(z)-S^*\bigr)>0$.
Then, for every $\lambda>0$,
\begin{equation}
    \TV\!\left(
    \omega_\lambda,
    \operatorname{Unif}(\calZ^*)
    \right)
    \le
    \frac{D-|\calZ^*|}{|\calZ^*|}e^{-\Delta/\lambda}.
\end{equation}
Moreover,
\begin{equation}\label{eq:ulimit}
    \left\|u_0^{(\lambda)}-u_0^{(0)}\right\|_\infty
    \le\frac{2\sigma_s}{\sqrt{\dt}}\,
    \frac{(D-|\calZ^*|)\,e^{-\Delta/\lambda}}{|\calZ^*|},
\end{equation}
where $u_0^{(0)}:=(\sigma_s/(|\calZ^*|\sqrt{\dt})) \sum_{z\in\calZ^*}z_0$. If all minimizers share the sign $s_i$ of $z_{0,i}$, then also
\begin{equation}
\bigl|u_{0,i}^{(\lambda)}-u_{0,i}^{(0)}\bigr|\le\frac{\sigma_s}{\sqrt{\dt}}\,\frac{D}{|\calZ^*|}\,e^{-\Delta_i/\lambda},
\end{equation}
with the coordinate gap $\Delta_i:=\min\{S(z)-S^*:z_{0,i}=-s_i\}\ge\Delta$.
\end{lemma}

\begin{proof}
After subtracting $S^*$, the weight outside $\calZ^*$ is $T_\lambda:=\sum_{z\notin\calZ^*}e^{-(S(z)-S^*)/\lambda}\le(D-|\calZ^*|)e^{-\Delta/\lambda}$, and the probability outside $\calZ^*$, which equals the total-variation distance to $\operatorname{Unif}(\calZ^*)$, is $T_\lambda/(|\calZ^*|+T_\lambda)\le T_\lambda/|\calZ^*|$. Applying
$|\E_\nu h-\E_{\nu'}h|\le2\TV(\nu,\nu')\|h\|_\infty$
with $h(z)=z_{0,i}$ proves~\eqref{eq:ulimit}. For the last claim, the set $\calZ_i^-:=\{z:z_{0,i}=-s_i\}$ has $D/2$ elements and excludes $\calZ^*$, and $u_{0,i}^{(\lambda)}=(\sigma_s/\sqrt{\dt})\,s_i(1-2\omega_\lambda(\calZ_i^-))$ with $\omega_\lambda(\calZ_i^-)\le(D/2)e^{-\Delta_i/\lambda}/|\calZ^*|$ similarly.
\end{proof}

\begin{corollary}\label{cor:dh}
Suppose that $S$ has a unique minimizer $z^*$ in $\calZ$. Then
$u_0^{(\lambda)}\to(\sigma_s/\sqrt{\dt})\,z_0^*$ as
$\lambda\to0$. With an exact cost oracle, D\"urr--H\o yer minimum
finding~\cite{durr96} solves Problem~\ref{prob:Q1} with
$\mathrm{tol}=0$ using $O(\sqrt D\log(1/\delta))$ expected oracle
queries, compared with $\Theta(D)$ cost evaluations for classical search.
\end{corollary}

\subsection{Normalization and implementation cost}\label{sec:results:cost}

The quantum bound~\eqref{eq:qcomplexity} depends on $a$, whereas the classical bound of Proposition~\ref{prop:classical} depends on $a'=a\,e^{S^*/\lambda}$, and the two can differ by many orders of magnitude when $S^*$ is large compared with $\lambda$. To remove this difference, the cost must be shifted by $S^*$ before the weights are encoded, which raises $a$ to $a'\ge|\calZ^*|/D$ by Remark~\ref{rem:range}. Since $S^*$ is unknown in general, we first compute the shift by solving Problem~\ref{prob:Q1}. If its output satisfies $S(\hat z)\le S^*+\mathrm{tol}$, then $\underline S:=\max\{0,S(\hat z)-\mathrm{tol}\}\in[0,S^*]$, so Remark~\ref{rem:range} applies and the shifted desirability lies in $[e^{-\mathrm{tol}/\lambda}a',a']$. Otherwise, negative shifted costs are set to zero, and the failure probability is included in $\delta$. Classical sampling is unaffected by this shift, since multiplying all weights by a constant changes neither its estimate nor its number of samples. Amplitude estimation, however, needs
\begin{equation}
O\Bigl(\frac{m}{\eps\sqrt a}\log\frac m\delta\Bigr)\ \ \text{and}\ \ O\Bigl(\sqrt D\log\frac1\delta+\frac{m}{\eps\sqrt{a'}}\log\frac m\delta\Bigr)
\end{equation}
queries without and with the shift, respectively.

Let $C_S$ be the number of $T$ gates in one query to $\calA$ or $\calA_i$ and $c_S$ the number of additions and multiplications in one classical simulation. The ratio $C_S/c_S$ compares operation counts only, since $T$ gates and classical operations do not have the same physical cost. Ignoring logarithmic factors, the modeled quantum cost is lower than the classical cost under the crossover condition
\begin{equation}\label{eq:crossover}
\frac{C_S}{c_S}
\min\left\{
\frac{m}{\eps\sqrt a},
\sqrt D+\frac{m}{\eps\sqrt{a'}}
\right\}
\lesssim
\min\left\{D,\;\frac{1}{\eps^2a'}\right\}.
\end{equation}
The left side of~\eqref{eq:crossover} is the modeled quantum cost, measured in classical simulations, without and with the shift. The right side is the classical cost, the smaller of evaluating all sequences and sampling as in Proposition~\ref{prop:classical}. Since $a\le a'$, the left side of~\eqref{eq:crossover} is at least $mX\,C_S/c_S$ with $X:=1/(\eps\sqrt{a'})$, while its right side is at most $\min\{D,X^2\}\le\sqrt D\,X$. The modeled gain, the ratio of the right to the left side of~\eqref{eq:crossover}, is therefore at most $\sqrt D\,X/(mX\,C_S/c_S)=\sqrt D\,c_S/(mC_S)$, and, whatever the accuracy and desirability,~\eqref{eq:crossover} can hold only if $D\gtrsim(mC_S/c_S)^2$, that is, $n_q\gtrsim2\log_2(mC_S/c_S)$. Therefore, the quadratic reduction in queries must first compensate for the much higher cost of each quantum query, and this happens only for a very large set of perturbation sequences.

\section{Numerical Validation}\label{sec:numerics}

We first test the predicted query scalings on a linear example. We then evaluate the crossover condition on this example and on a nonlinear pendulum.\footnote{The code and a supplementary note on the query cost are available at \repourl.}

\emph{Single integrator:} Here $f\equiv0$, $G=I_2$, $m=2$, $q(x)=c_g\|x-x_g\|_2^2+ c_o e^{-\|x-x_o\|_2^2/(2r^2)}$, and $\phi(x)=c_T\|x-x_g\|_2^2$, with $\bar x=(0,0)$, $x_g=(1,0)$, $x_o=(0.5,0)$, $r=0.15$, $N=4$, $\dt=0.25$, $\sigma_s=0.4$, $R=I_2$, and $(c_g,c_T,c_o)=(1,10,20)$, so $\lambda=\sigma_s^2=0.16$ satisfies A3, $n_q=8$, and $D=256$.

Evaluating all sign sequences gives $S^*=3.058$, $\Delta=0.035$, and two symmetric minimizers with $z_0=(+1,\pm1)$, so $u_0^{(0)}=(0.800,0)$. At $\lambda=0.16$, $b_1/a=0.99999$, $b_2/a=0.500$, and $u_0^{(\lambda)}=(0.800,0)$ to three decimals. This agreement is explained by the last part of Lemma~\ref{lem:conc}, not by~\eqref{eq:ulimit}, which gives no useful bound here, since $(D-|\calZ^*|)e^{-\Delta/\lambda}/|\calZ^*|\approx102$. Both minimizers have $z_{0,1}=+1$ with coordinate gap $\Delta_1=1.90$, which bounds the first-coordinate deviation by $7.2\times10^{-4}$ (the actual value is $1.8\times10^{-5}$), and the cost is symmetric in $x_2$, so $u_{0,2}^{(\lambda)}=0$ for every $\lambda$. The desirabilities are $a=7.1\times10^{-11}$ and $a'=0.0141$. Over $2{,}000$ runs simulated with exact Grover measurement probabilities, D\"urr--H\o yer minimum finding returned a true minimizer with a mean of $11.8$ calls to $O_\gamma$ and $17.5$ classical cost evaluations to update the threshold.

\begin{figure}[t]
\centering
\includegraphics[width=0.97\columnwidth]
{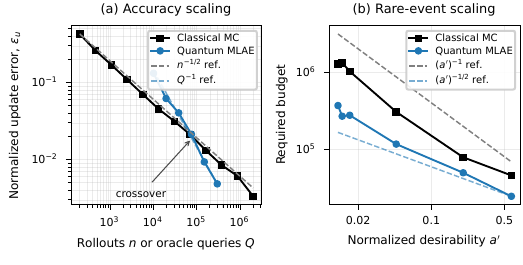}
\caption{Single integrator, with normalized update error $\eps_u:=\|\hat u_0-u_0^{(\lambda)}\|_\infty/(\sigma_s/\sqrt{\dt})$, shown as its median over trials. (a) $\eps_u$ versus the number of simulations for classical sampling (MC), labeled rollouts $n$, or quantum queries $Q$ at $a'=0.0141$. (b) Required budget, that is, simulations or queries, for $\eps_u=5\times10^{-3}$ as $a'$ varies through $\lambda\in\{0.08,0.16,0.5,1.5,5,20\}$. Dashed lines show the predicted rates.}
\label{fig:qae}
\end{figure}

\begin{table}[t]
\caption{Modeled cost of one quantum query.}
\label{tab:resource}
\centering
\setlength{\tabcolsep}{4pt}\renewcommand{\arraystretch}{0.92}
\begin{tabular}{lrr}
\hline
Resource & Integrator & Pendulum, $n_q=28$\\
\hline
Register width $w$ & $29$ & $32$\\
Toffoli, forward & $62{,}582$ & $550{,}656$\\
$T$-count, with uncomputation & $5.0\times10^5$ & $4.4\times10^6$\\
Logical qubits & $\approx270$ & $\approx1{,}213$\\
Classical operations & $\approx124$ & $\approx725$\\
$C_S/c_S$ & $\approx4.0\times10^3$ & $\approx6.1\times10^3$\\
\hline
\end{tabular}
\end{table}

Figure~\ref{fig:qae} compares classical sampling with the median of three runs of maximum-likelihood amplitude estimation (MLAE)~\cite{suzuki20}, an amplitude-estimation algorithm different from the one in the proof of Theorem~\ref{thm:qae}. Here $Q$ counts all oracle calls, and $S^*$ is known. The curves in Fig.~\ref{fig:qae}(a) cross near $\eps_u=2\times10^{-2}$, and at the last point MLAE uses $2.97\times10^5$ queries versus about $1.2\times10^6$ simulations.
In Fig.~\ref{fig:qae}(b), the fitted classical and quantum exponents are $-0.92$ and $-0.68$. The quantum exponent is steeper than the small-$a'$ limit $-1/2$ since the amplitude-estimation error scales as $\sqrt{p(1-p)}/Q$, and the ratios $b_i/a$ change with $\lambda$. Accounting for both effects predicts an exponent of $-0.67$.

\begin{table}[t]
\caption{Pendulum family and condition~\eqref{eq:crossover} at $w=32$.}
\label{tab:pend}
\centering
\setlength{\tabcolsep}{4pt}\renewcommand{\arraystretch}{0.92}
\begin{tabular}{rccc}
\hline
$n_q$ & $a'$ & $\eps$ satisfying~\eqref{eq:crossover} & Max.\ gain$^\dagger$\\
\hline
24 & $1.7\times10^{-4}$ & none & $0.65$\\
26 & $1.2\times10^{-4}$ & $[0.0086,0.015]$ & $1.3$\\
28 & $8.9\times10^{-5}$ & $[0.0025,0.017]$ & $2.6$\\
30 & $1.1\times10^{-4}$ & $[0.00055,0.015]$ & $5.2$\\
\hline
\multicolumn{4}{l}{\footnotesize$^\dagger$Largest ratio of the right to the left side of~\eqref{eq:crossover}.}
\end{tabular}
\end{table}

\emph{Nonlinear pendulum:} We consider a pendulum with angle $x_1$ and angular velocity $x_2$, controlled by a torque input, so that $f(x)=(x_2,-\omega_0^2\sin x_1)$ with $\omega_0^2=9.81$, $G=(0,1)^{\!\top}$, and $m=1$ in~\eqref{eq:sde}. The costs $q(x)=1+\cos x_1+0.1x_2^2$ and $\phi(x)=10(1+\cos x_1)+x_2^2$ penalize deviations from the upright position $x_1=\pi$, and the initial state is $\bar x=(\pi-0.4,0)$. We use $\dt=0.05$, $\sigma_s=1$, and $R=1$, so $\lambda=1$ satisfies A3. For each $N=8,\ldots,30$, we evaluate all $D=2^N$ sign sequences, so the horizon $N\dt$ grows with $N$. The normalized desirability decreases from $a'=1.57\times10^{-3}$ at $N=16$ to about $10^{-4}$ at larger $N$ (Table~\ref{tab:pend}).

\emph{Resource model:} Table~\ref{tab:resource} estimates the cost of one query by counting Toffoli gates, each of which requires four $T$ gates~\cite{jones13}. For a fixed-point register width of $w$ bits, we count $w$ Toffoli gates per addition and $w^2$ per multiplication, evaluate the nonlinear functions by piecewise cubic polynomials without the selection of the polynomial piece, and include uncomputation. To check that these widths suffice, we simulated the fixed-point computation of the cost with $\kappa=5$. The cost error stayed below the bound $\lambda\eta/4$ of Lemma~\ref{lem:prec}, by factors between $1.7$ and $7.9$, for all integrator sequences at $\eps=5\times10^{-3}$ and for $2\times10^5$ sampled pendulum sequences at the accuracy of maximum gain in Table~\ref{tab:pend} ($n_q\ge26$). These simulations omit the errors from approximating the nonlinear functions and compiling the rotation, so they support but do not prove the hypotheses of Lemma~\ref{lem:prec}. With the bit count given after Lemma~\ref{lem:prec}, angle rounding alone needs $17$ bits for $a'=0.0141$ at $\eps=5\times10^{-3}$, and $31$ for the unshifted $a$. With both examples at $w=29$, the pendulum's $C_S/c_S$ is about $24\%$ larger because of its nonlinear functions, and the larger $w$ needed near the upright position adds a similar increase.

For the crossover condition~\eqref{eq:crossover}, the integrator has $D=256$, well below the necessary size $(mC_S/c_S)^2\approx6.5\times10^7$, so~\eqref{eq:crossover} fails at every accuracy. The necessary size $(C_S/c_S)^2\approx3.7\times10^7$ gives $n_q\ge26$, where~\eqref{eq:crossover} first holds, and the maximum gain then grows as $\sqrt D$. With worst-case bounds and unit constants,~\eqref{eq:crossover} is only an order-of-magnitude comparison. In Fig.~\ref{fig:qae}, MLAE needs about $80\,m/(\eps\sqrt{a'})$ queries and classical sampling about $0.46/(\eps^2a')$ simulations. With these constants, the necessary number of sign sequences grows by about $(80/0.46)^2\approx2^{15}$, that is, by about $15$ qubits, and the constants vary between problem instances. As almost all weight falls on sequences with $z_{0,1}=+1$ ($b_1/a\ge0.96$ for $12\le n_q\le28$), the sampled first inputs barely vary, so the errors of the two averages largely cancel in~\eqref{eq:ratio}, and classical sampling is much more accurate than its worst-case bound. At $n_q=16$, it reaches $\eps_u\approx2\times10^{-3}$ with $10^4$ simulations, whereas MLAE needs $8\times10^5$ queries for $4\times10^{-3}$.

\section{Conclusion}\label{sec:conclusion}

We encoded the MPPI update as a ratio of measurement probabilities with a specified precision and proved that amplitude estimation needs quadratically fewer queries than classical sampling. As $\lambda\to0$ with a unique minimum-cost sequence, the update converges to the first input of this sequence, which quantum minimum finding locates with quadratically fewer queries. For the cost model considered here, the crossover condition can hold only if $D\gtrsim(mC_S/c_S)^2$. Future work includes circuits with verified precision, variance-adaptive estimators, and closed-loop guarantees.

\bibliographystyle{IEEEtran}
\bibliography{refs}

@article{kappen05, author={H. J. Kappen}, title={Path integrals and symmetry breaking for optimal control theory}, journal={J. Stat. Mech.}, volume={2005}, number={11}, pages={P11011}, year={2005}}

@article{todorov09, author={E. Todorov}, title={Efficient computation of optimal actions}, journal={Proc. Natl. Acad. Sci. USA}, volume={106}, number={28}, pages={11478--11483}, year={2009}}

@article{williams17, author={G. Williams and A. Aldrich and E. A. Theodorou}, title={Model predictive path integral control: From theory to parallel computation}, journal={J. Guid. Control Dyn.}, volume={40}, number={2}, pages={344--357}, year={2017}}

@article{williams18, author={G. Williams and P. Drews and B. Goldfain and J. M. Rehg and E. A. Theodorou}, title={Information-theoretic model predictive control: Theory and applications to autonomous driving}, journal={IEEE Trans. Robot.}, volume={34}, number={6}, pages={1603--1622}, year={2018}}

@incollection{bhmt02, author={G. Brassard and P. H{\o}yer and M. Mosca and A. Tapp}, title={Quantum amplitude amplification and estimation}, booktitle={Quantum Computation and Information}, series={Contemp. Math.}, volume={305}, pages={53--74}, year={2002}}

@article{montanaro15, author={A. Montanaro}, title={Quantum speedup of {Monte Carlo} methods}, journal={Proc. R. Soc. A}, volume={471}, number={2181}, pages={20150301}, year={2015}}

@inproceedings{yoon22, author={H.-J. Yoon and C. Tao and H. Kim and N. Hovakimyan and P. Voulgaris}, title={Sampling complexity of path integral methods for trajectory optimization}, booktitle={Proc. Amer. Control Conf. (ACC)}, pages={3482--3487}, year={2022}}

@article{patil24, author={A. Patil and G. A. Hanasusanto and T. Tanaka}, title={Discrete-time stochastic {LQR} via path integral control and its sample complexity analysis}, journal={IEEE Control Syst. Lett.}, volume={8}, pages={1595--1600}, year={2024}}

@inproceedings{novara24, author={C. Novara and M. Boggio and D. Volpe}, title={A quantum optimization approach to nonlinear model predictive control}, booktitle={Proc. Eur. Control Conf. (ECC)}, pages={2810--2815}, year={2025}}

@misc{weinberg26, author={A. I. Weinberg}, title={{QuantFPFlow}: Quantum amplitude estimation for {Fokker--Planck} policy optimisation in continuous reinforcement learning}, year={2026}, note={arXiv:2605.16429}}

@book{kloedenplaten, author={P. E. Kloeden and E. Platen}, title={Numerical Solution of Stochastic Differential Equations}, publisher={Springer}, address={Berlin}, year={1992}}

@book{nielsenchuang, author={M. A. Nielsen and I. L. Chuang}, title={Quantum Computation and Quantum Information}, edition={{10th anniversary}}, publisher={Cambridge Univ. Press}, address={Cambridge, U.K.}, year={2010}}

@article{herbert21, author={S. Herbert}, title={No quantum speedup with {Grover--Rudolph} state preparation for quantum {Monte Carlo} integration}, journal={Phys. Rev. E}, volume={103}, pages={063302}, year={2021}}

@inproceedings{ar20, author={S. Aaronson and P. Rall}, title={Quantum approximate counting, simplified}, booktitle={Proc. SIAM Symp. Simplicity Algorithms}, pages={24--32}, year={2020}}

@inproceedings{nayakwu99, author={A. Nayak and F. Wu}, title={The quantum query complexity of approximating the median and related statistics}, booktitle={Proc. ACM Symp. Theory Comput.}, pages={384--393}, year={1999}}

@misc{durr96, author={C. D{\"u}rr and P. H{\o}yer}, title={A quantum algorithm for finding the minimum}, year={1996}, note={arXiv:quant-ph/9607014}}

@article{suzuki20, author={Y. Suzuki and S. Uno and R. Raymond and T. Tanaka and T. Onodera and N. Yamamoto}, title={Amplitude estimation without phase estimation}, journal={Quantum Inf. Process.}, volume={19}, pages={75}, year={2020}}

@article{vedral96, author={V. Vedral and A. Barenco and A. Ekert}, title={Quantum networks for elementary arithmetic operations}, journal={Phys. Rev. A}, volume={54}, number={1}, pages={147--153}, year={1996}}

@misc{haner18, author={T. H{\"a}ner and M. Roetteler and K. M. Svore}, title={Optimizing quantum circuits for arithmetic}, year={2018}, note={arXiv:1805.12445}}

@article{babbush21, author={R. Babbush and J. R. McClean and M. Newman and C. Gidney and S. Boixo and H. Neven}, title={Focus beyond quadratic speedups for error-corrected quantum advantage}, journal={PRX Quantum}, volume={2}, pages={010103}, year={2021}}

@article{jones13, author={C. Jones}, title={Low-overhead constructions for the fault-tolerant {Toffoli} gate}, journal={Phys. Rev. A}, volume={87}, pages={022328}, year={2013}}

@article{deboer05, author={P.-T. de Boer and D. P. Kroese and S. Mannor and R. Y. Rubinstein}, title={A tutorial on the cross-entropy method}, journal={Ann. Oper. Res.}, volume={134}, pages={19--67}, year={2005}}

@article{sadegh98, author={P. Sadegh and J. C. Spall}, title={Optimal random perturbations for stochastic approximation using a simultaneous perturbation gradient approximation}, journal={IEEE Trans. Autom. Control}, volume={43}, number={10}, pages={1480--1484}, year={1998}}

\end{document}